\documentclass[aps,pra,onecolumn,nofootinbib,superscriptaddress]{revtex4}
\usepackage{amsmath,amssymb,amsthm}
\usepackage{amsfonts}
\usepackage{amssymb}
\usepackage[final]{graphicx}
\usepackage{color}
\usepackage[a4paper, total={6in, 8in}]{geometry}
\usepackage{ulem}

\newtheorem{theorem}{Theorem}

\newtheorem{prop}[theorem]{Proposition}

\usepackage[title]{appendix}
\usepackage{float}
\usepackage[caption = false]{subfig}
\DeclareMathOperator{\Tr}{Tr}

\begin{document}

\title{Entropic Bounds as Uncertainty Measure of Unitary Operators}

\author{Jesni Shamsul Shaari}
\affiliation{Faculty of Science, International Islamic University Malaysia (IIUM),
Jalan Sultan Ahmad Shah, Bandar Indera Mahkota, 25200 Kuantan, Pahang, Malaysia}
\author{Rinie N. M. Nasir}
\affiliation{Faculty of Science, International Islamic University Malaysia (IIUM),
Jalan Sultan Ahmad Shah, Bandar Indera Mahkota, 25200 Kuantan, Pahang, Malaysia}

\author{Stefano Mancini}
\affiliation{School of Science \& Technology, University of Camerino, I-62032 Camerino, Italy}
\affiliation{ INFN Sezione di Perugia, I-06123 Perugia, Italy}
\date{\today}

\begin{abstract}
We reformulate the notion of uncertainty of pairs of unitary operators within the context of guessing games and derive an entropic uncertainty relation for a pair of such operators. We show how distinguishable operators are compatible while maximal incompatibility of unitary operators can be connected to bases for some subspace of operators which are mutually unbiased.  
\end{abstract}


\maketitle

\section{Introduction}
The notion of uncertainty, or incompatibility has almost become synonymous with quantum theory. The most  immediate case would be the uncertainty for a pair of noncommuting observables captured in the famous inequality \cite{scho} for observables $A$ and $B$ given by $\Delta A\Delta B\geq \langle\psi|[AB-BA]|\psi\rangle|/2$ with the uncertainties in observables $O$ given as standard deviations, $\Delta O=\sqrt{|\langle\psi|O^2|\psi\rangle|-|\langle\psi|O|\psi\rangle|^2}$, for a quantum state $|\psi\rangle$.

This has to some extent motivated the study of uncertainty of pairs of unitary operators which obey a specific commutation relation \cite{maasar}. A generalisation to include arbitrary pairs of unitary operators with tighter bounds was later reported \cite{bagchi}. In these works, the uncertainty of an operator, $U$, was defined in an analogous way to the uncertainties of observables, in terms of variances, $\Delta U=\sqrt{1-|\langle\psi|U|\psi\rangle|^2}$. The uncertainty defined this way can be more precisely understood as the uncertainty in state $|\psi\rangle$ associated to $U$ in the context of the overlap between the state  $|\psi\rangle$ and the evolved state $U|\psi\rangle$. Extending this to a pair of operators, $U, V$, an uncertainty relation would supposedly reflect the distinguishability of two distinct unitary state evolutions \cite{bagchi}.

However, notwithstanding the motivations mentioned in the Refs. \cite{maasar,bagchi} of defining $\Delta U$ as above, the meaning of \textit{uncertainty of unitary operators} in reference to such a definition can be somewhat ambiguous. We shall refer to a simple example for our argument, the Pauli unitary operators, $\sigma_x$ and $\sigma_z$ operating on the two dimensional Hilbert space, $\mathcal{H}_2$. The uncertainty relation $\Delta\sigma_x^2+\Delta\sigma_z^2\ge 1$ then was reinterpreted in \cite{maasar} as an uncertainty relation for the Mach-Zender interferometer. However, in terms of unitary operators, it can be a little unclear what this greater than zero value would imply as such a pair of unitary operators actually can result in distinguishable states for a specific input. Namely, for input  $|0\rangle$ (in terms of the computational basis), for which it is easy to note that $|\langle 0|\sigma_x\sigma_z|0\rangle|=0$. Thus the distinguishability of the evolved states is perfect. This is hardly surprising as the Pauli operators are actually orthogonal and can be distinguished perfectly.

Therefore, we shall consider in this work a more `operational' notion for the uncertainty of unitary operators. Intuitively, we would like some definition such that the notion would reflect the uncertainty one can have when \textit{testing} such an operator. Testing itself can be understood as a means to ascertain the nature of action the operator has on some quantum state. In other words, if one can \textit{test} perfectly such an operator, then the uncertainty would be zero.  Extending to the case of testing a pair of operators, we require our formulation to give zero uncertainty ONLY when the \textit{tester} has zero uncertainty with respect to both the operators while perfectly distinguishing the operators. 
The framework of testing operators is well established in the context of Process Positive Operator Value Measurement (PPOVM) \cite{ziman} or quantum tester \cite{bisio, sedlak}. This can be understood as analogous to what POVMs are to quantum states.

As for the notion of uncertainty, we will resort to the most obvious Shannon entropy, defined as $H(x)=-\sum p(x)\log{p(x)}$ for $x$ being the outcome of the tester occurring with probability $p(x)$. We note that our work is essentially motivated  by the use of entropic uncertainty relations for observables of finite dimensional Hilbert spaces \cite{D,maassen, step, coles}. The entropic uncertainty relation for such observables was first proposed by \cite{D} and later refined by \cite{maassen} to eventually replace the older relations based on standard deviations. It is more often described in terms of guessing games; i.e. one party, say Alice, provides another, say Bob with some arbitrary state for which the latter would measure in either one of two measurement basis. Bob then informs Alice of his choice of measurement and she then guesses the result of his measurement. Her guesses for both measurements cannot be free of uncertainty if the measurements are incompatible. Inspired by that, we will formulate the notion of uncertainty of unitary operators pairs in our work, within the context of some guessing game.    
The main results of our work highlighted as:
\begin{enumerate}
\item uncertainty of unitary operators simply reflects the distinguishability of the operators. The incompatibility of observables can be connected to the distinguishability of unitary operators as follows: if unitary operators $U$ and $V$ are (not perfectly) distinguishable, then the measurement operators defined by $\{U^\dagger|\chi_i\rangle\langle\chi_i|U\}$ and $\{V^\dagger|\chi_i\rangle\langle\chi_i|V\}$ are incompatible for any projective measurement, $\{|\chi_i\rangle\langle\chi_i|\}$. Unless stated otherwise, we assume throughout the text, projective measurements such that $\langle\chi_i|\chi_j\rangle=\delta_{ij}$ and $\{|\chi_i\rangle\}$ is a complete basis for the $d$ dimensional Hilbert space.
\item if the uncertainty saturates the maximal entropic bound (i.e. $\log{d}$) for some pair $U$ and $V$, then the measurement operators defined by $\{U^\dagger|\chi_i\rangle\langle\chi_i|U\}$ and $\{V^\dagger|\chi_i\rangle\langle\chi_i|V\}$ are maximally incompatible. In the special case where each of these operators comes from two distinct unitary bases and the maximal entropic bound is reached for any pairs thus formed, then the bases are mutually unbiased \cite{scott,jssarxiv,arxiv}.
\end{enumerate}
We should note the difference (though related) work of Ref.\cite{annals} where an entropic relation was defined for a pair of unitary testers. In the mentioned reference, an uncertainty relation using Shannon entropy was defined for a tester pair in an analogous way to that of a pair observables. In simple terms, it reflected the uncertainty when considering a pair of testers for some (common) unitary tested. This is markedly different from the current work which defines uncertainty of results when a pair of unitary operators are tested by a common tester. As such it complements Ref.\cite{annals} and, we shall see later how some of the resulting propositions relate to those of Ref.\cite{annals}.

The outline of the work is as follows. We begin with a brief review of testers, mutually unbiased bases (MUB) and mutually unbiased unitary bases (MUUB). We then commit to a thorough study of uncertainty of pairs of operators using testers with a (pure) $d$-dimensional input state and projective measurements. This would include propositions related to the distinguishability of the operators and another which connects maximal incompatibility to MUUB. We further present a short section on using testers with maximally entangled input states as well as those using POVMs for measurement purposes before concluding the work. 
\section{MUBs, MUUBs and Testers}
\noindent Given  the measurement bases $\mathbb{V}=\{|\mathbb{V}_i\rangle\langle\mathbb{V}_i|\}$ and $\mathbb{W}=\{|\mathbb{W}_i\rangle\langle\mathbb{W}_i|\}$, with $\{|\mathbb{V}_i\rangle\}$ and $\{|\mathbb{W}_j\rangle\}$ being distinct bases for $\mathcal{H}_d$, the notion of incompatibility of measurements  is captured through the well known uncertainty relation \cite{maassen}, 
\begin{eqnarray}\label{ents}
H(\mathbb{V})+H(\mathbb{W})\ge -\log{\max_{i,j}|\langle \mathbb{V}_i|\mathbb{W}_j\rangle}|^2
\end{eqnarray}
with $H(\cdot)$ the Shannon entropy of the outcomes. The bound is saturated if and only if $\{|\mathbb{V}_i\rangle\}$ and $\{|\mathbb{W}_j\rangle\}$ are mutually unbiased to one another \cite{step,coles}; i.e, 
\begin{eqnarray}
|\langle\mathbb{W}_i|\mathbb{V}_j\rangle|=1/\sqrt{d}, ~\forall i,j,=0,\ldots, d-1.
\end{eqnarray}
These bases, referred to as mutually unbiased bases (MUB), are described by having equiprobable transition between states from one basis to the other \cite{durt}.

An analogous structure to MUB for the space, $\mathcal{L}(\mathcal{H}_d)$, i.e. the space of linear operators on $\mathcal{H}_d$  was studied in ref.\cite{scott,jssarxiv,arxiv} as mutually unbiased unitary basis (MUUB). The definition carries similarly, the notion of inner product (in this case, the Hilbert-Schmidt inner-product);  
two distinct orthogonal basis, $\{P_0,...,P_{D-1}\}$ and $\{Q_0,...,Q_{D-1}\}$ comprising of unitary transformations for some $D$-dimensional subspace of the space $\mathcal{L}(\mathcal{H}_d)$ are sets of MUUB if 
\begin{eqnarray}\label{MUUBdef}
\left|\text{Tr}(P_i^{\dagger}Q_j)\right |^2=\kappa~~,~~\forall i,j,=0,\ldots, D-1
\end{eqnarray}
for some constant, $\kappa$. This constant takes on value $1$ and $d$ for the $d^2$ and $d$ dimensional subspaces of $\mathcal{L}(\mathcal{H}_d)$ respectively \cite{jssarxiv,arxiv}. 

The discrimination of operators can be done using testers. Physically, a tester consist of some known quantum state which acts as a probe that undergoes the action of some (unknown) operator and subsequently, measurements are made to determine how the state has evolved; thus one may infer properties of the operator. Testers can then be described as a pair, $T=(\Psi,\{F_i\})$ where $\Psi$ is the known quantum state and $\{F_i\}$ is a POVM for the measurement purpose. In principle, $\Psi$ can be any quantum state, although it should preferably be pure; thus avoiding further uncertainties in determining an operator. 
We note that this is the same description used in Ref.\cite{annals} and is sufficient for our purpose. We refer the interested reader to Ref.\cite{ziman} for a more precise definition in terms of PPOVM.


\section{Entropic Bounds for Pairs of Unitary Testing}
\noindent In what follows, we shall make use of testers with a pure $d$-dimensional input state, $|\psi\rangle\in \mathcal{H}_d$ and measurement operators $\{|\chi_i\rangle\langle\chi_i|\}$ assumed to be orthogonal projectors.  As noted earlier, we formulate the notion of uncertainty of unitary operators in terms of testers within the framework of a guessing game as follows:\\
\\
\textit {Consider a pair of unitary operators $V,W\in SU(d)$. Let Alice prepare a tester, $\mathcal{T}=(|\psi\rangle, \{|\chi_i\rangle\langle\chi_i|\})$, to distinguish between the two operators and submits the tester to Bob. Bob then would randomly test either $V$ or $W$ and then informs Alice of the operator he used. Alice then needs to guess Bob's results.} 
\\\\
Assuming the unitary operators were chosen with equal probability, and the tester used as above, Alice's average uncertainty of Bob's result would be given as
\begin{eqnarray}\label{un}
H(\mathcal{T}|V)+H(\mathcal{T}|W)=-\sum_i|\langle\chi_i|V|\psi\rangle|^2\log{|\langle\chi_i|V|\psi\rangle|^2}\\\nonumber-\sum_i|\langle\chi_i|W|\psi\rangle|^2\log{|\langle\chi_i|W|\psi\rangle|^2}
\end{eqnarray}
The right hand side of the above equation can be understood equivalently as the average uncertainty of  making measurements using the measurement operators $\mathcal{V}=\{V^\dagger|\chi_i\rangle\langle\chi_i|V\}$ and $\mathcal{W}=\{W^\dagger|\chi_i\rangle\langle\chi_i|W\}$ given the state $|\psi\rangle$.
Hence 
\begin{eqnarray}\label{un}
H(\mathcal{T}|V)+H(\mathcal{T}|W)=H(\mathcal{V}||\psi\rangle)+H(\mathcal{W}||\psi\rangle)
\end{eqnarray}
and we can immediately apply the well-known entropic bound for measurement operators (independent of $|\psi\rangle$)
and thus we write an inequality for eq.(\ref{un}) as 
\begin{eqnarray}\label{ent}
H(\{|\chi_i\rangle\langle\chi_i|\}|V)+H(\{|\chi_i\rangle\langle\chi_i|\}|W)\geq -\log{\max_{i,j}{|\langle\chi_i|WV^\dagger|\chi_j\rangle|^2} }.
\end{eqnarray}
We see that this relation is dependent on the measurement operator of the tester (for any possible input of the tester,  thus we write $\{|\chi_i\rangle\langle\chi_i|\}$ for notation rather than $\mathcal{T}$). This is as it should be, i.e. it reflects the physical fact that \textit{any} uncertainty is borne out of the measurements and not due to the deterministic nature of unitary evolutions. Hence an uncertainty relation in the form of entropic inequality makes sense only when referred to the tester involved.

The entropic relation reflect Alice's uncertainty in guessing Bob's measurement results. Incompatibility of unitary operators is thus interpreted as Alice's inability to prepare a tester to distinguish between the operators which would give zero uncertainty for both $V$ and $W$.
Notwithstanding this, the relation has another clear physical meaning. It is  in fact, the uncertainty relation for a common measurement operator type used by the two parties independently, whose reference frames are related by a certain unitary relation. Let Alice and Bob independently measure the observable, $\{|\chi_i\rangle\langle \chi_i|\}$, in their respective locality for a given state sent by Alice to Bob (or identical states by a third party to both). If their reference frames are rotated to one another, such that any state prepared by Alice is viewed by Bob as rotated by the operator $WV^\dagger$, then the entropic uncertainty relation describes the uncertainty of their measurement outcomes.


For a given measurement operator, $\{|\chi_i\rangle\langle\chi_i|\}$, one can obviously derive an infinite number of testers by having different inputs. However, only a specific subset of those, if exists, may achieve the minimal entropic value, $-\log{\max_{i,j}{|\langle\chi_i|WV^\dagger|\chi_j\rangle|^2} }$. We refer to such testers as \textit{saturating testers}. 

Note that for any unitary operator pair $V$ and $W$, one can always have zero uncertainty, i.e. when one considers a tester such that the measurement $\{|\chi_i\rangle\langle\chi_i|\}$ projects onto the eigenvectors of $WV^\dagger$ for some specific input. This is quite obvious and we refer to testers with such measurement operators as \textit{trivial testers for the unitary operators pair $V$ and $W$}. These trivial testers unfortunately have no real operational use. It tells us nothing about the difference between $V$ and $W$, beyond being a passive operation for some input of the tester. In what follows, we shall exclude trivial testers from our considerations.


\subsection{The Minimal Bound; Perfect Discrimination}

\begin{prop}
Consider a pair of unitary operators, $V$ and $W$. Excluding trivial testers, the uncertainty relation gives the trivial (zero) bound, i.e.
\begin{eqnarray}\label{eqnT}
H(\{|\chi_i\rangle\langle\chi_i|\}|V)+H(\{|\chi_i\rangle\langle\chi_i|\}|W))=0
\end{eqnarray}
for some $\{|\chi_i\rangle\langle\chi_i|\}$ if and only if they are perfectly distinguishable
\end{prop}
\begin{proof}

A pair of unitary $V$ and $W$ being perfectly distinguishable implies the overlap between $V|\chi_j\rangle$ and $W|\chi_j\rangle$ for some $W|\chi_j\rangle$ is given by  $|\langle\chi_j|WV^\dagger |\chi_j\rangle|=0$. Hence, there exists a $\{|\chi_i\rangle\langle\chi_i|\}$, such that $|\langle\chi_j|WV^\dagger |\chi_i\rangle|=1$. This results in $H(\{|\chi_i\rangle\langle\chi_i|\}|V)+H(\{|\chi_i\rangle\langle\chi_i|\}|W)=0$.
\\
\\
Proving the other way, it is obvious that a pair of unitary operators, $V$ and $W$, which gives the minimal (zero) bound implies $H(\{|\chi_i\rangle\langle\chi_i|\}|V)=H(\{|\chi_i\rangle\langle\chi_i|\}|W)=0$ and $|\langle\chi_i|WV^\dagger |\chi_j\rangle|=1$ for some $i\neq j$ (having excluded the trivial tester), thus $|\langle\chi_i|V^\dagger W|\chi_i\rangle|=0$ which implies $V$ and $W$ can be perfectly distinguished. 

\end{proof}
\noindent This of course means that any two unitary operators, $V$ and $W$ such that $\Tr{(V^\dagger W)}=0$; i.e orthogonal, and thus distinguishable would saturate the minimal bound for some $\{|\chi_i\rangle\langle\chi_i|\}$. Such pairs would imply the existence of some measurement operator, $\{|\chi_i\rangle\langle\chi_j|\}$, such that   $\{V^\dagger|\chi_i\rangle\langle\chi_i|V\}$ and $\{W^\dagger|\chi_i\rangle\langle\chi_i|W\}$ are compatible.

As an example, the unitary operators $P,Q\in SU(d)$ (as defined in Ref.\cite{maasar} with $j,k=-[d/2],...,[(d-1)/2]$) such that $PQ=QPe^{2i\pi/d}$ given as
 \begin{eqnarray}
 P=\sum_{j}e^{i2\pi j/d}|a_j\rangle\langle a_j|~,~ Q=\sum_{k}^{}e^{-i2\pi k/d}|b_k\rangle\langle b_k|
 \end{eqnarray} 
 with the two orthonormal basis $\{|a_j\rangle\}$ and $\{|b_k\rangle\}$ related based by the Discrete Fourier transform, 
 \begin{eqnarray}
 |b_k\rangle=\sum_{j}e^{i2\pi jk/d}|a_j\rangle
 \end{eqnarray}
are orthogonal (it can easily be shown that $  \Tr{(PQ^\dagger)}=0$), therefore distinguishable and  $H(P)+H(Q)=0$. 
 If we consider a basis of unitary operators (for some space of operators), the orthogonality of the operators imply that they are all in fact compatible, in the sense discussed above. This is notwithstanding the fact that they may be non-commuting. Commutation (or otherwise) of operators has no bearing on its compatibility within our notion of uncertainty of unitary operators and is only meaningful in the context of operators as observables. Owing to the different frameworks subscribed when considering the roles of the operators, the Pauli operators, say, for $SU(2)$ which are non-commuting are thus non-compatible in the context of observables though compatible as unitary operators.

Let us consider how this can be related to the trivial entropic bound of ref.\cite{annals}. Let $T_1$ and $T_2$ be two testers and if they both satisfy eq.(\ref{eqnT}), then, rewriting $U_1$ for $V$ and $U_2$ for $W$  
we have, 
\begin{eqnarray}
H(T_1|U_1)+H(T_1|U_2)+H(T_2|U_1)+H(T_2|U_2)=0
\end{eqnarray}
which implies
\begin{eqnarray}
H(T_1|U_i)+H(T_2|U_i)=0, ~i=1,2.
\end{eqnarray}
By further considering the notion of complete set of testers \cite{annals} and assuming the above to also hold for the set of operators $\{U_1,U_2,...,U_d\}$, the proposition 5 related to trivial bounds of ref.\cite{annals} follows.

\subsection{The Maximal Bound; MUUBs}
\noindent Given the measurement of our tester results in any one of $d$ possible outcomes, the maximal bound, $\log{d}$, can be achieved for a pair of unitary operators $V$ and $W$ if $WV^\dagger|\chi_j\rangle$ and $|\chi_k\rangle$ are mutually unbiased for all $j,k$. 

However, when this holds true for unitary operators pairwise taken from two distinct bases, it has some interesting implications on the inner-product of the operators. We summarise it in the following proposition:
  \begin{prop}\label{prop2}
Let $\mathfrak{V}=\{V_i\}$ and $\mathfrak{W}=\{W_j\}$ be two distinct bases for some $d$-dimensional subspace for, $\mathcal{L}(\mathcal{H}_d)$, the space of linear operators acting on a $d$-dimensional Hilbert space. $\mathfrak{V}$ and $\mathfrak{W}$ are mutually unbiased, i.e. $|\Tr{(V_n^\dagger W_m)}|=\sqrt{d}$ if for any $m,n$, the unitary pair $V_n$ and $W_m$ saturates the maximal bound for some tester, $T_{m,n}$.
\end{prop}
\begin{proof}
Let $H(\{|\chi_i\rangle\langle\chi_i|\}|V_n)+H(\{|\chi_i\rangle\langle\chi_i|\}|W_m))=-\log{1/d}$, i.e. $V_n$ and $W_m$ saturate the maximal bound for any $m,n$ given some tester,  $T_{m,n}$. This condition is met if and only if $\{V_n^\dagger|\chi_i\rangle\}$ and $\{W_m^\dagger|\chi_i\rangle\}$ are mutually unbiased; i.e. $\forall i,j, |\langle\chi_i|W_mV_n^\dagger|\chi_j\rangle|=1/\sqrt{d}$.
The trace of the operator $W_mV_n^\dagger$ can be written as 
\begin{eqnarray}
\Tr{(W_mV_n^\dagger)}=\sum_i\langle\chi_i|W_mV_n^\dagger|\chi_i\rangle
\end{eqnarray}
 and 
 \begin{eqnarray}\label{inequal1}
 \sum_i|\langle\chi_i|W_mV_n^\dagger|\chi_i\rangle|\geq |\sum_i\langle\chi_i|W_mV_n^\dagger|\chi_i\rangle|=|\Tr{(W_mV_n^\dagger)}|
\Rightarrow |\Tr{(W_mV_n^\dagger)}|\leq \sqrt{d}.
\end{eqnarray}
We will now demonstrate how $|\Tr{(W_mV_n^\dagger)}|$ is in fact necessarily equal to $\sqrt{d}$.
 Let us write $W_m=\sum\alpha_iV_i$ and thus $|\Tr{(W_mV_n^\dagger)}|^2=|\alpha_n|^2d^2$. However note that 
   \begin{eqnarray}
   |\Tr{(W_mW_m^\dagger)}|=d(\sum_i|\alpha_i|^2)=d\Rightarrow \sum_i|\alpha_i|^2=1.  
    \end{eqnarray}
    Therefore,
  \begin{eqnarray}\label{equality1}
\sum_k|\Tr{(W_mV_k^\dagger)}|^2=d^2\sum|\alpha_k|^2=d^2
\end{eqnarray}
From  (\ref{inequal1}) earlier, we have $|\Tr{(W_mV_n^\dagger)}|\leq \sqrt{d}$, for which squaring both sides of the inequality gives $|\Tr{(W_mV_n^\dagger)}|^2\leq d$ and $\sum_k|\Tr{(W_mV_k^\dagger)}|^2\leq d^2$. The requirement of equation (\ref{equality1}) implies $\forall k, |\Tr{(W_mV_k^\dagger)}|^2= d$ or $|\Tr{(W_mV_n^\dagger)}|=\sqrt{d}$.
 \end{proof}
 The proposition in a nutshell describes how any operator pair formed with elements each taken from two unitary operator basis, which are maximally incompatible implies that the two operator bases are in fact mutually unbiased. 

Our notation of tester, $T_{m,n}$, in the above proposition reflects the possibility of having testers of differing inputs dependent on the different values of $m,n$, despite the common measurement operators. Let us consider how we can define the tester. Let $W_mV_n^\dagger|\chi_i\rangle=|\zeta_i^{(m,n)}\rangle$ where $|\langle \chi_i|\zeta_i^{(m,n)}\rangle|=1/\sqrt{d}$. We further denote $|\psi_i^{(n)}\rangle=V_n^\dagger|\chi_i\rangle$, implying $V_n|\psi_i^{(n)}\rangle=|\chi_i\rangle$. Hence, $W_m|\psi_i^{(n)}\rangle=W_mV_n^\dagger|\chi_i\rangle=|\zeta_i^{(m,n)}\rangle$. Thus, if we pick a tester, 
$T_{m,n}=(|\psi_i^{(n)}\rangle,\{|\chi_i\rangle\langle\chi_i |\})$, the saturation of the maximal entropic bound, 
$H(T_{m,n}|W_m)+H(T_{m,n}|V_n)=\log{d}$, holds with the uncertainties in testing $V_n$ and $W_m$ being zero and $\log{d}$ respectively. 

Using a similar argument to the above, we can obviously consider another tester, say, $T_{m,n}'$, such that, its input is $|\psi_i^{(m)}\rangle=W_m^\dagger|\chi_i\rangle$ and thus the uncertainty in testing $V_n$ would be maximal and $W_m$ would be zero instead. We can therefore write
\begin{eqnarray}
H(T_{m,n}|W_m)+H(T_{m,n}'|W_m)=
H(T_{m,n}|V_n)+H(T_{m,n}'|V_n)=\log{d}.
\end{eqnarray}
Considering the case where the above holds for any pair of testers, $T_{m,n}$ and $T_{m,n}'$, each coming from two different complete set of testers,  proposition 6 of Ref.\cite{annals} follows immediately.
 \subsection{The case for SU(2)}
 \noindent We consider here a special case for operators of SU(2), some specific unitary operators acting on the 2-dimensional Hilbert space. Its simple structure would allow us to explicitly see the features of the uncertainty for selected pairs of operators using testers as described earlier.
 
We consider two pairs of operators; the identity, $I$ with the Pauli operator, $\sigma_y$, and the identity with $(I-i \sigma_y)/\sqrt{2}$. The projective measurement operators in both cases are given by measurement $\{|\chi_i\rangle\langle\chi_i|\}$ with 
\begin{eqnarray}|\chi_1\rangle=\cos{\theta}|0\rangle+e^{i\phi}\sin{\theta}|1\rangle~,~
|\chi_2\rangle=-\sin{\theta}|0\rangle+e^{i\phi}\cos{\theta}|1\rangle
\end{eqnarray}
with $\theta,\phi\in[0,\pi]$.  Now, the overlaps between $\sigma_y|\chi_i\rangle$ and $|\chi_j\rangle$ for $i,j=1,2$ are given as,
  \begin{eqnarray}  
|\langle \chi_i |I  \sigma_y |\chi_j \rangle|^2=
  \begin{cases}
                                 \sin^2 2 \theta \sin^2  \phi & \text{, $i=j$} \\
                                  \cos^4 \theta+ 2 \cos^2  \theta \cos^2 2 \phi \sin^2 \theta+\sin^4 \theta  & \text{, $i \neq j$} \\
  
  \end{cases}
\end{eqnarray}
and that for $I$ and $\Omega=(I-i\sigma_y)/\sqrt{2}$ 
\begin{eqnarray}
  |\langle \chi_i |I ((I-i \sigma_y)/\sqrt{2})^{\dagger} |\chi_j \rangle|^2=
  \begin{cases}
                                   [1+\sin^2 (2 \theta)\sin^2 (\phi)]/2   & \text{, $i=j$} \\
                                    [1-\sin^2 (2 \theta)\sin^2 (\phi)]/2 & \text{, $i \neq j$} \\
  
  \end{cases}
\end{eqnarray}

\noindent We plot below in Fig. 1 (a) and Fig. 2 (a) respectively, the graphs  for $\max_{i,j}{|\langle\chi_i|I\sigma_y|\chi_j\rangle|^2}$ and $\max_{i,j}{|\langle\chi_i|(I-i \sigma_y)^{\dagger}/\sqrt{2}|\chi_j\rangle|^2}$ as functions of $\theta,\phi$. Additionally, Fig. 1 (b) and Fig. 2 (b) respectively illustrates $|\langle\chi_i|I\sigma_y|\chi_i\rangle|^2$ and $|\langle\chi_i|(I-i \sigma_y)^{\dagger}/\sqrt{2}|\chi_i\rangle|^2$, highlighting the overlap between the states evolved from the different operators in each pair. 

\begin{figure}[h]
\subfloat[]{\includegraphics[width=0.45\textwidth]{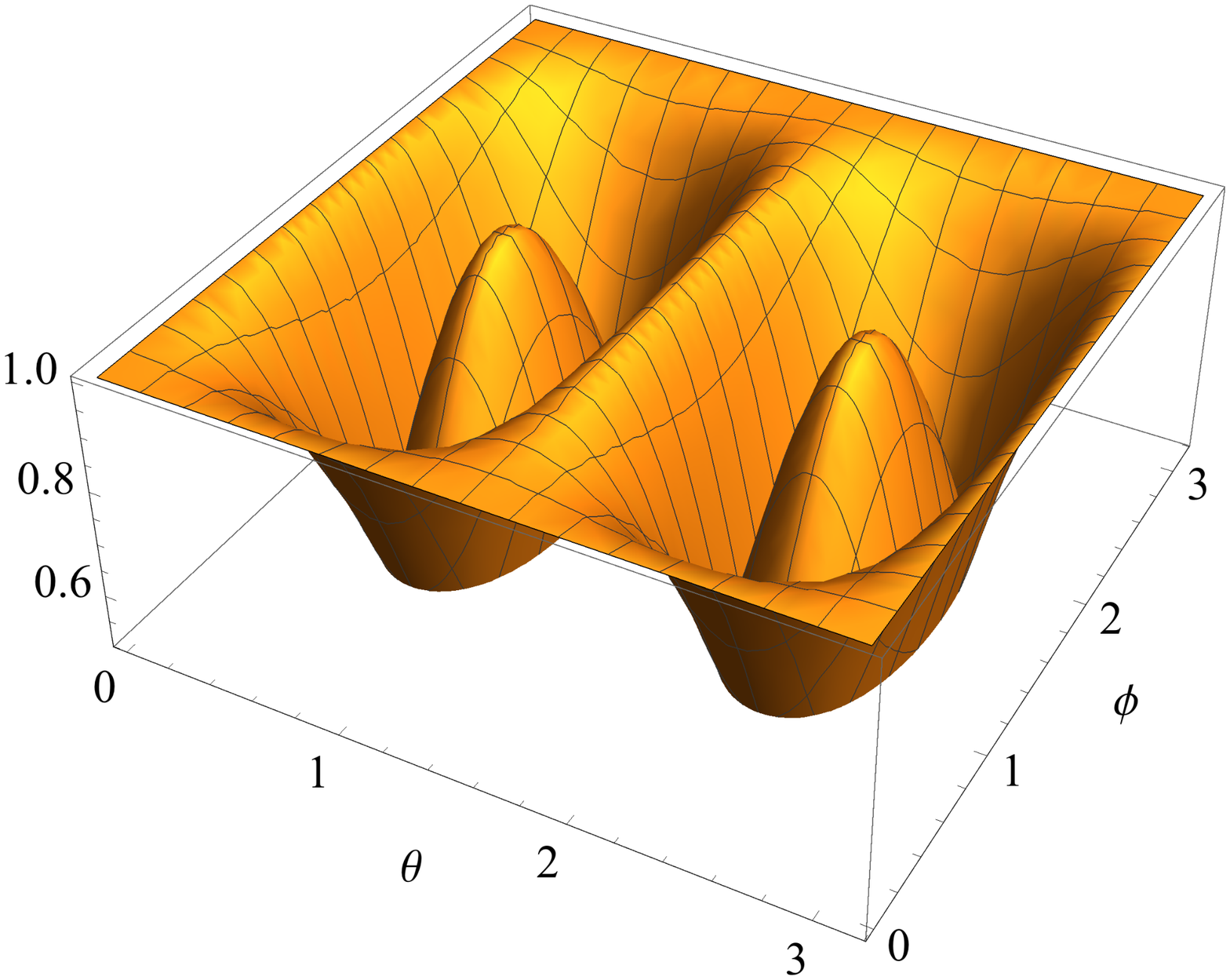}} 
\subfloat[]{\includegraphics[width=0.45\textwidth]{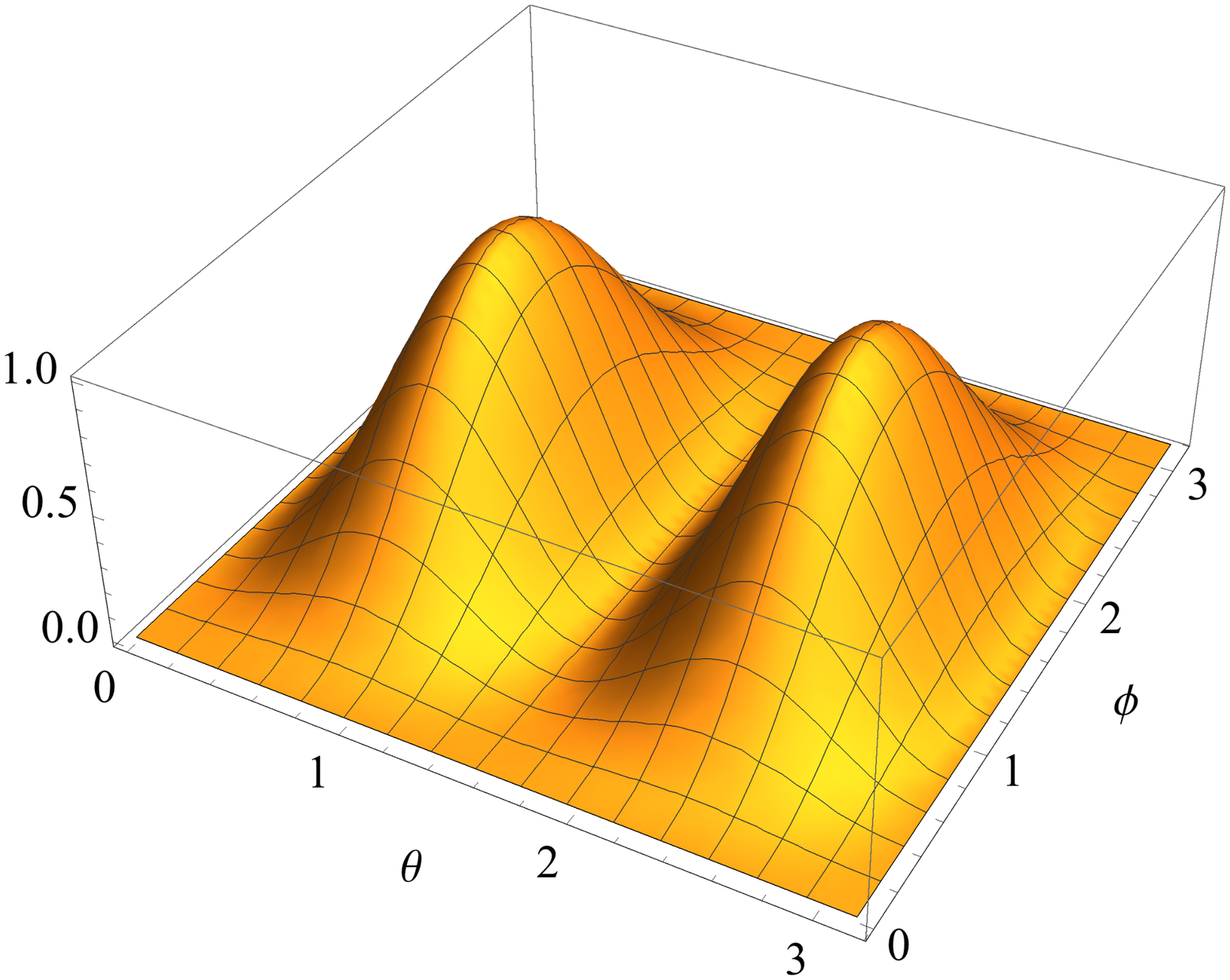}}
\caption{ The plots (a) for  $\max_{i,j}{|\langle\chi_i|I\sigma_y|\chi_j\rangle|^2}$ and (b) for $|\langle\chi_i|I\sigma_y|\chi_i\rangle|^2$ for any $i=1,2$,  are against the angles $\theta$ and $\phi$.} 
\label{some example}
\end{figure}

\begin{figure}[h]
\subfloat[]{\includegraphics[width=0.45\textwidth]{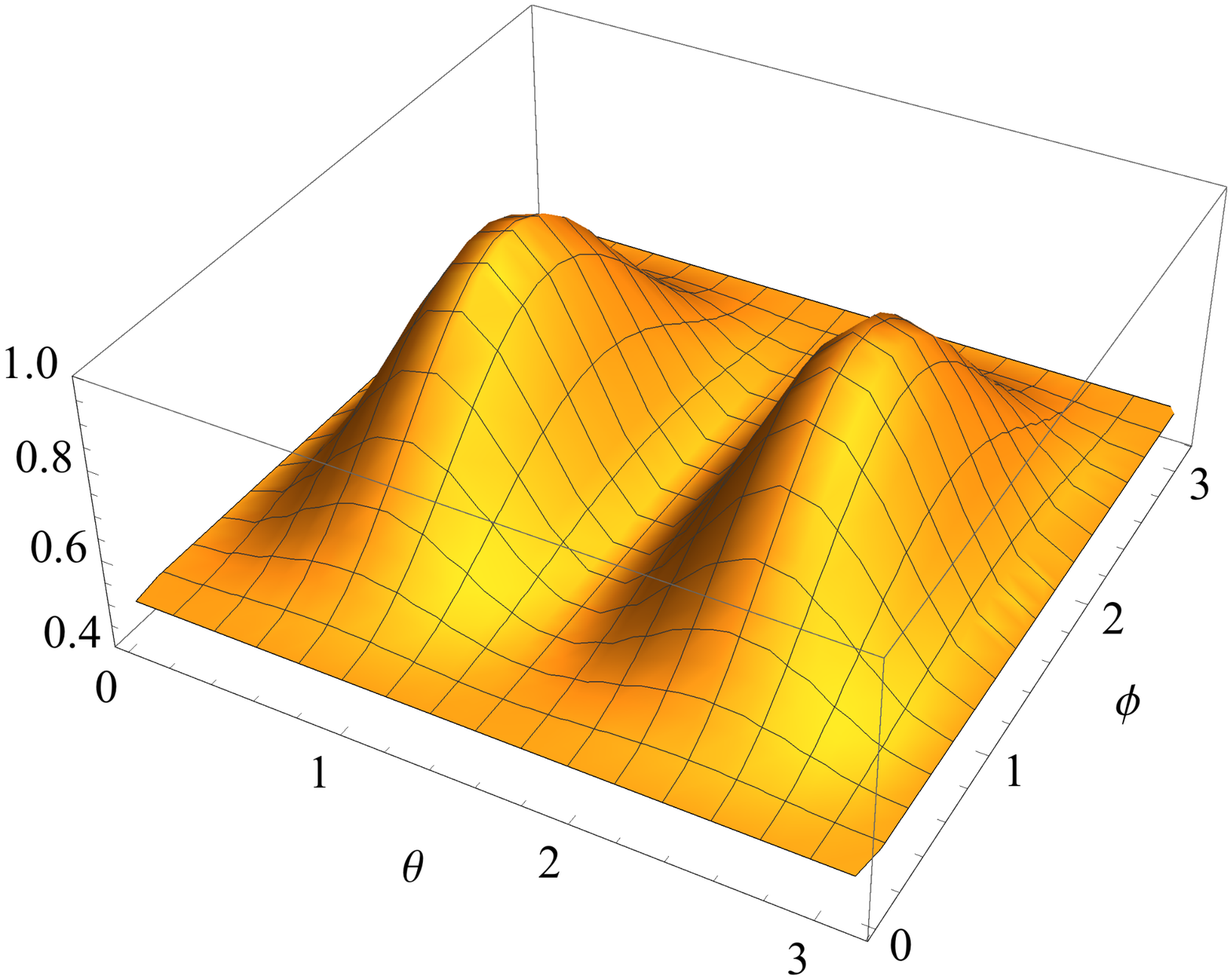}} 
\subfloat[]{\includegraphics[width=0.45\textwidth]{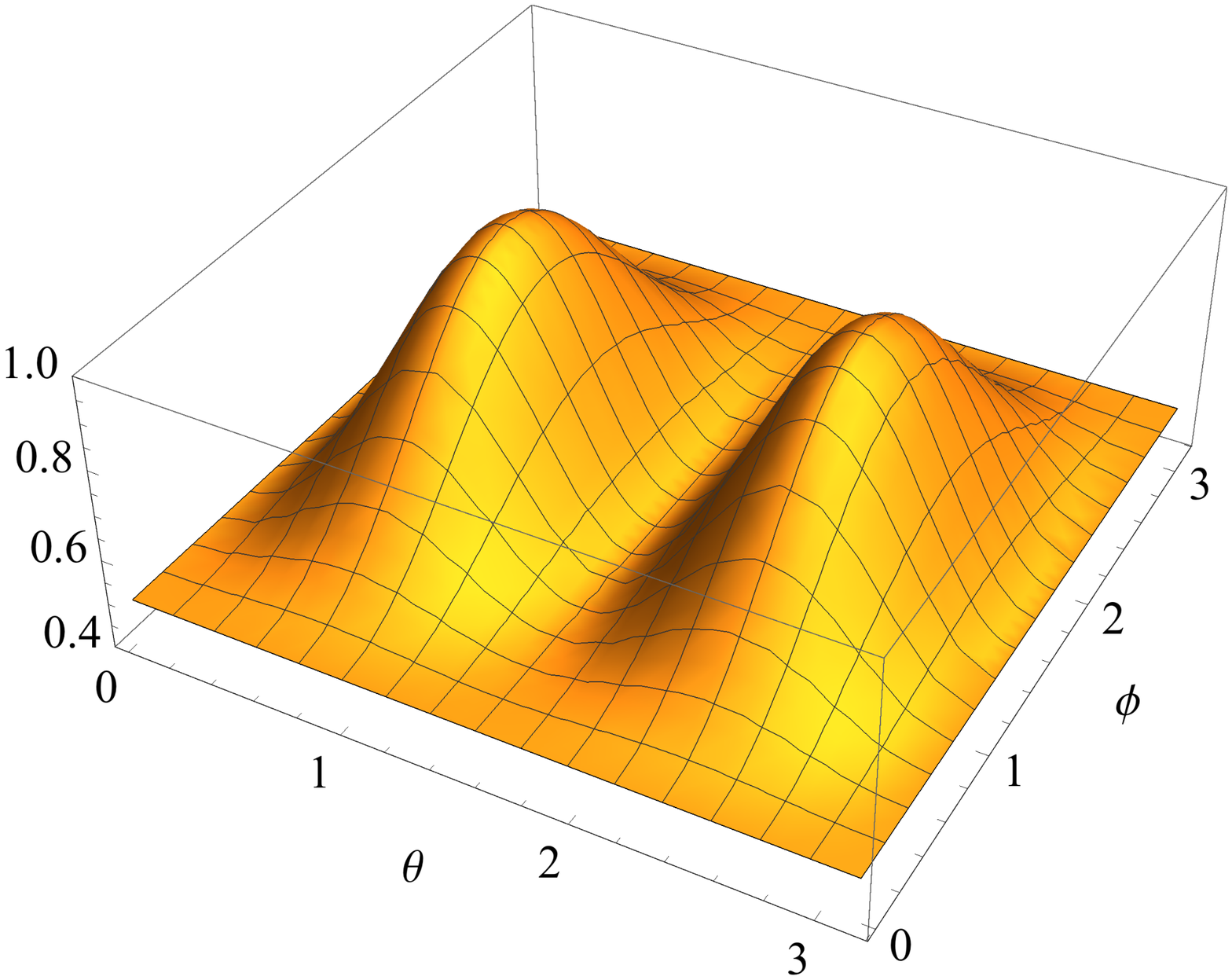}}
\caption{The plots (a) for  $\max_{i,j}{|\langle\chi_i|(I-i \sigma_y)^{\dagger}/\sqrt{2}|\chi_j\rangle|^2}$ and  (b) for $|\langle\chi_i|(I-i \sigma_y)^{\dagger}/\sqrt{2}|\chi_i\rangle|^2$ for any $i=1,2$, are against the angles $\theta$ and $\phi$.  In both the plots here, one can note the minimal value as $0.5$. The graphs are essentially identical as $\max_{i,j}{|\langle\chi_i|(I-i \sigma_y)^{\dagger}/\sqrt{2}|\chi_j\rangle|^2}=|\langle\chi_i|(I-i \sigma_y)^{\dagger}/\sqrt{2}|\chi_i\rangle|^2$.}
\label{some example}
\end{figure}
The graph of Fig1(a) demonstrates that $\max_{i,j}{|\langle\chi_i|I\sigma_y|\chi_j\rangle|^2}$, achieves the maximal value 1 (hence zero uncertainty) at the `edges' as well as at the centres. The edges correspond to projectors that project onto the states $\cos{\theta}|0\rangle+\sin{\theta}|1\rangle, \forall \theta$, i.e. those lying on the equator of the Poincare sphere. The edges also include the case for $\theta=0,\pi$ for any possible value of $\phi$; which are really the states of the computational basis. In these cases, one can in fact use any of these states as the input state to form the saturating tester. 

We can observe that $\max_{i,j}{|\langle\chi_i|I\sigma_y|\chi_j\rangle|^2}$, decreases in value (increases in uncertainty) up to 0.5 before taking a turn to increase to 1 at the centres. This increase is contributed by the scenario where $\max_{i,j}{|\langle\chi_i|I\sigma_y|\chi_j\rangle|^2}=|\langle\chi_i|I\sigma_y|\chi_i\rangle|^2$  for any $i=1,2$; as depicted in Fig.1 (b), and reflects the distinguishability of the states $\sigma_y|\chi_i\rangle$ and $I|\chi_i\rangle$.	
The maximum value at the centre correspond to the eigenvectors of $I\sigma_y$ (trivial tester) and we exclude this from our considerations.

In the context of Bob's rotated frame with respect to Alice's, this trivial bound can be interpreted as Alice measuring in the operators $\{|\chi_i\rangle\langle\chi_i|\}$ while Bob measures $\{\sigma_y|\chi_i\rangle\langle\chi_i|\sigma_y\}$. But as $|\langle\chi_j|\sigma_y|\chi_i\rangle|=1$ for some $j\neq i$, thus $|\langle\chi_i|\sigma_y|\chi_i\rangle|=0$ for all other $i$ and given the fact that we are only dealing with qubits here, $\{\sigma_y|\chi_i\rangle\langle\chi_i|\sigma_y\}=\{|\chi_i\rangle\langle\chi_i|\}$.

The graph of  Fig. 2(a) on the other hand demonstrates that, for any state, $|\chi_i\rangle$, lying on the equator of the Poincare sphere, $\max_{i,j}{|\langle\chi_i|(I-i \sigma_y)^{\dagger}/\sqrt{2}|\chi_j\rangle|^2}= 1/2$. The overlap between $(I-i\sigma_y)^\dagger/\sqrt{2}|\chi_i\rangle$ and $I|\chi_i\rangle$ for any $i=1,2$ as shown in Fig. 2(b) has its its minimal value of $1/2$ for those states and becomes less distinguishable otherwise. The graphs are essentially identical due to the fact that the contribution for $\max_{i,j}{|\langle\chi_i|(I-i \sigma_y)^{\dagger}/\sqrt{2}|\chi_j\rangle|^2}$ comes solely from $|\langle\chi_i|(I-i \sigma_y)^{\dagger}/\sqrt{2}|\chi_i\rangle|^2$ for any $i=1,2$. We also note in this case, a peak at the centre which corresponds to the zero uncertainty scenario happens for the case of trivial testers. 
It is instructive to note that, similar graphs can be drawn for the operator pairs $\sigma_y$ and $(I-i \sigma_y)/\sqrt{2}$, $\sigma_y$ and $(I+i \sigma_y)/\sqrt{2}$  as well as  $I$ and $(I+i \sigma_y)/\sqrt{2}$, demonstrating the connection between the saturation of the maximal bound to MUUBs. The bases $\{I,\sigma_y\}$ and $\{(I-i \sigma_y)/\sqrt{2},(I+i \sigma_y)/\sqrt{2}\}$ for the subspace of operators of SU(2) are in fact mutually unbiased. 
\section{Generalisations}
\noindent An immediate extension of the earlier sections would be the case where the tester uses a bipartite maximally entangled state (MES), $|\Psi\rangle\in\mathcal{H}_d\otimes\mathcal{H}_d$, as input and the measurement operators, $\{|\nu_i\rangle\langle\nu_i|\}$, would project onto some MES as well. This is to avoid any contribution of uncertainty due to measurements that project an unitarily evolved maximally entangled input onto non maximally entangled states\footnote{a measurement which projects a MES onto separable states would result in maximal uncertainty independently of the unitary operators involved}. This does not modify the preceeding results much beyond the following:
\begin{enumerate}
\item Trivial testers do not exist. This is due to the fact that the operator $WV^\dagger\otimes I_d$ for any $W,V\in SU(d)$ is an eigenoperator for a MES only if $WV^\dagger$ is some multiple of $I_d$ and only thus can $|\langle \nu_i|WV^\dagger\otimes I_d|\nu_i\rangle|$ be equal to unity.
\item Proposition \ref{prop2} remains except that $\mathfrak{V}=\{V_i\}$ and $\mathfrak{W}=\{W_j\}$ are now two distinct bases for the $d^2$-dimensional space of operators and $|\Tr{(W_mV_n^\dagger)}|=1$, for any $W_m\in\mathfrak{W},V_n\in\mathfrak{V}$ inline with the MUUB definition for the space. 
\end{enumerate}
As these are rather straightforward, the explicit demonstration is referred to only in the Appendix section for the interested reader.

Another immediate direction of generalisation would be the straightforward case for using testers with POVMs for measurement purposes. Similar to the earlier section, we will see how it reduces to uncertainty relations for measurements by two parties with rotated POVMs. Let $\rho$ be the input and $\{ M_k\}$ as the POVM for tester $T$. The average uncertainty now gives,
\begin{align}
H(V|T)+H(W|T)=-\sum_{k} \text{Tr}[M_k V^{\dagger} \rho V] \log  \text{Tr}[M_k V^{\dagger} \rho V]
\nonumber\\
-\sum_{k} \text{Tr}[M_k W^{\dagger} \rho W] \log  \text{Tr}[M_k W^{\dagger} \rho W]\nonumber\\
=-\sum_{k} \text{Tr}[VM_k V^{\dagger} \rho] \log  \text{Tr}[VM_k V^{\dagger} \rho]
\nonumber\\
-\sum_{k} \text{Tr}[WM_k W^{\dagger} \rho] \log  \text{Tr}[WM_k W^{\dagger} \rho]
\end{align} 
Hence, for any $\rho$, this becomes the entropic bound for POVMs \cite{sankhya}, $\mathcal{M}_V=\{M_i^{(V)}~|~M_i^{(V)}=VM_i V^{\dagger}\}$ and $\mathcal{M}_W=\{M_j^{(W)}~|~M_j^{(W)}=WM_j W^{\dagger}\}$,
\begin{align}
H(\mathcal{M}_V)+H(\mathcal{M}_W)\ge -2 \log{\max_{i,j} \|\sqrt{M_i^{(V)}}\sqrt{M_j^{(W)}}} \|.
\end{align} 
where $\| \cdot \|$ is the operator norm.

\section{Conclusion}

\noindent The notion of uncertainty pairs of unitary operators in terms of variances derived quite straightforwardly from those related to operator observables is certainly not without its charm. Aiming to sideline the possible ambiguity in such a formulation, especially when dealing with pairs that deterministically evolved states which are completely distinguishable, we have revised such notion of uncertainty to one based on entropy using testers. 
Formulated within the framework of a guessing game, the uncertainty in testing a unitary operator pair,  $V$ and $W$, reflects the inability of one party to prepare a tester from which conclusive statements an be made for both the operators.

We immediately see how operators which are completely distinguishable; i.e. orthogonal are in fact compatible in the sense that it saturates the minimal entropic bound. It essentially implies the existence of some measurement operator for which the evolution of its eigenbasis  under two compatible unitary operators are in fact perfectly distinguishable. Hence, any two operators from any common basis for some subspace of the space of operators are compatible. This also means that the (generalised) Pauli operators are actually compatible irrespective of the matter of commutation of operators.   It is thus instructive to determine what a pair of unitary operator represents; i.e. deterministic evolution or as observables  of a quantum system in order to discuss notions of incompatibility.

On the other extreme, we observe that operators derived from distinct bases for some subspace of the space of operators saturating the maximal bound are in fact mutually unbiased to one another. This is true for both the $d$ as well as the $d^2$ dimensional subspaces $\mathcal{L}(\mathcal{H}_d)$.

The framework we have described in terms of the uncertainty relation with testers provides a more obvious and practical meaning to the idea of `compatibility' of the operators. It rightfully delegates the focus of uncertainty to measurement in the testers themselves and not the unitary operation which are deterministic.  Its physical implication can also be captured in terms of independent measurements made by the parties Alice and Bob for which their measurement operators are connected by some unitary relation. In a nutshell, incompatibility of unitary operators should be viewed in terms of its action on measurement operators, and any uncertainties hence, are derived from subsequent measurements. 

It's worth noting the relation the current work has to Ref.\cite{annals}, where the latter considers the entropic uncertainty relation involved when a unitary is tested by one of two choices of testers. Our propositions in the current work, as we have demonstrated, can actually be understood as primitives for some of the results of Ref.\cite{annals}.


\section{Acknowledgement}

\noindent J. S. S. and R. M. N would like to acknowledge financial support under the project FRGS19-030-0639 from the Ministry of Higher Education's Fundamental Research Grant Scheme and the University's Research Management Centre (RMC) for their support and facilities provided. J. S. S would also like to extend a very special thanks to Su'aidah (RMC, Kuantan) for her kind assistance and support. S. M. acknowledges the financial support of the Future and Emerging Technologies (FET) programme, within the Horizon-2020 Programme of the European Commission, under the FET-Open grant agreement QUARTET, number 862644.


\section{Appendix}

\begin{prop}\label{K0}
Consider $\mathfrak{S}=\{S_i\}$ and $\mathfrak{R}=\{R_j\}$ be two distinct bases for the $d^2$-dimensional space of operators acting on the $d$-dimensional Hilbert space. A pair of unitary operators, $S_p$ and $R_q$ saturating the maximal bound for any $p,q$, implies that $\mathfrak{S}$ and $\mathfrak{R}$ are mutually unbiased, i.e. $| \text{Tr} (S_pR_{q}^{\dagger})|=1$.
\end{prop}

\begin{proof}
Let $H(S_p)+H(R_q)=-\log{1/d^2}$,  i.e. $S_p$ and $R_q$ saturates the maximal bound for any $p,q$. saturating the maximal bound. This condition is met if and only if $\{R_q^{\dagger}  \otimes \mathcal{I} |\nu_i\rangle\}$ and $\{S_p  \otimes \mathcal{I}|\nu_j\rangle\}$ are mutually unbiased; i.e. $\forall i,j, |\langle\nu_i|S_p R_q^{\dagger}  \otimes \mathcal{I}|\nu_j\rangle|=1/d$. Note that,
\begin{align}
\Tr{(S_p R_q^{\dagger}  \otimes \mathcal{I})}&=\Tr{(S_p R_q^{\dagger})} \Tr{(\mathcal{I})} \nonumber\\
&= d \Tr{(S_p R_q^{\dagger})}
\end{align}
\noindent
Now, consider that the trace of the operator $S_p R_q^{\dagger} \otimes \mathcal{I}$ can be written as 
\begin{eqnarray}
\sum_i\langle\nu_i|S_p R_q^{\dagger}  \otimes \mathcal{I} |\nu_i\rangle=d \Tr{(S_p R_q^{\dagger})}
\end{eqnarray}
 and 
 \begin{align}\label{inequal}
 \sum_i|\langle\nu_i|S_p R_q^{\dagger}  \otimes \mathcal{I}|\nu_i\rangle| \geq   |\sum_i\langle\nu_i|S_p R_q^{\dagger} \otimes \mathcal{I} |\nu_i\rangle| =|\Tr{(S_p R_q^{\dagger} \otimes \mathcal{I}})| = d| \Tr{(S_p R_q^{\dagger})}| \Rightarrow |\Tr{(S_p R_q^{\dagger} )}|\leq 1.
\end{align}
We will now demonstrate how $|\Tr{(S_p R_q^{\dagger} )}|$ is in fact necessarily equal to $1$.\\\\
\noindent Let us write $S_p=\sum \alpha_i R_i$ and thus $|\Tr{(S_p R_q^{\dagger} )}|^2=d^2|\alpha_q|^2$. Note that 
   \begin{eqnarray}
   |\Tr{(S_p S_p^{\dagger})}|=d (\sum_i|\alpha_i|^2)=d \Rightarrow \sum_i|\alpha_i|^2=1.  
    \end{eqnarray}
    Therefore,
  \begin{eqnarray}\label{equality}
\sum_k|\Tr{(S_p R_k^{\dagger} )}|^2=d^2 \sum|\alpha_k|^2=d^2
\end{eqnarray}
From  (\ref{inequal}) earlier, we have $|\Tr{(S_p R_q^{\dagger} )}|\leq 1$, for which squaring both sides of the inequality gives $|\Tr{(S_p R_q^{\dagger} )}|^2\leq 1$ and $\sum_k|\Tr{(S_p R_k^{\dagger}  )}|^2\leq d^2$. The requirement of equation (\ref{equality}) implies $\forall k, |\Tr{(S_p R_k^{\dagger})}|^2= 1$. 
 \end{proof}



\begin{thebibliography}{99}
\bibitem{scho} E. Schr\"{o}dinger, Proc. Pruss. Acad. Sci. Phys. Math. Sec. XIX,
293 (1930).
\bibitem{maasar} S. Massar and P. Spindel, Phys. Rev. Lett. 100, 190401 (2008).
\bibitem{bagchi} S. Bagchi and A. K. Pati, 
Phys. Rev. A 94, 042104 (2016).
\bibitem{ziman} M. Ziman, 
Phys. Rev. A 77, 062112 (2008).
\bibitem{bisio} A. Bisio, G. Chiribella, G. M. D'Ariano, P. Perinotti, Acta Phys. Slovaca 61, 3 (2011).
\bibitem{sedlak} M. Sedlak, D. Reitzner, G. Chiribella and M. Ziman, 
 Phys. Rev. A 93, 052323 (2016).
\bibitem{D} D. Deutsch, 
Phys. Rev. Lett. 50, 631 (1983).
\bibitem{maassen} H. Maassen and J. Uffink, 
Phys. Rev. Lett. 60, 1103 (1988).
\bibitem{step} S. Wehner and A. Winter, 
New J. Phys. 12, 025009 (2010).
\bibitem{coles} P. J. Coles, M. Berta, M. Tomamichel and S. Wehner, 
Rev. Mod. Phys. 89, 015002 (2017).
\bibitem{scott} A. J. Scott, J. Phys. A 41, 055308 (2008).
\bibitem{jssarxiv} J. S. Shaari, R. N. M. Nasir and S. Mancini, 
Phys. Rev. A 94, 052328 (2016).
\bibitem{arxiv} R. N. M. Nasir, J. S. Shaari, S. Mancini, Quantum Inf. Process., 18: 178 (2019).
\bibitem{annals} J. S. Shaari and S. Mancini, Ann. Phys. (N. Y.), 412, 168043 (2020).
\bibitem{durt} T. Durt, B. Englert, Bengtsson, K. \.Zyczkowski, 
Int. J. Quantum Inf., 8, (2010). 
\bibitem{sankhya} M. Krishna and K. R. Parthasarathy (2002), Sankhya 64 (3), 842.
\end{thebibliography}
\end{document}